\documentclass[%
reprint,superscriptaddress,
amsmath,amssymb,
aps,
]{revtex4-1}
\usepackage[
colorlinks,
linkcolor = blue,
citecolor = blue,
urlcolor = blue]{hyperref}

\usepackage{amsmath}
\usepackage{amssymb}
\usepackage{amsfonts}
\usepackage{enumerate}
\usepackage{mathrsfs}
\usepackage{graphicx}
\usepackage{epstopdf}
\usepackage{enumerate}
\usepackage{changepage}
\usepackage{bm}
\usepackage{multirow}
\usepackage{lipsum}
\usepackage{booktabs}
\usepackage{array}
\usepackage{amsthm}
\usepackage{float}
\usepackage{dcolumn}
\usepackage{bm}

\newcommand{\ucite}[1]{\textsuperscript{\cite{#1}}}

\begin{document}
	
	\preprint{APS/123-QED}
	
	\title{Strongest nonlocal sets with minimum cardinality in tripartite systems}
	
	\author{Xiao-Fan Zhen }
	\affiliation{School of Mathematical Sciences, Hebei Normal University, Shijiazhuang, 050024, China}

	\author{Mao-Sheng Li}
\affiliation{ School of Mathematics,
	South China University of Technology, Guangzhou
	510641,  China}

	\author{Hui-Juan Zuo }
	\email{huijuanzuo@163.com}
	
	\affiliation{School of Mathematical Sciences, Hebei Normal University, Shijiazhuang, 050024, China}%
	
	
	\begin{abstract}
Strong nonlocality, proposed by Halder {\it et al}. [\href{https://doi.org/10.1103/PhysRevLett.122.040403}{Phys. Rev. Lett. \textbf{122}, 040403 (2019)}], is a stronger manifestation than quantum nonlocality. Subsequently, Shi {\it et al}. presented the concept of the strongest nonlocality [\href{https://doi.org/10.22331/q-2022-01-05-619}{Quantum \textbf{6}, 619 (2022)}].
		Recently, Li and Wang [\href{https://doi.org/10.22331/q-2023-09-07-1101}{Quantum \textbf{7}, 1101 (2023)}] posed the conjecture about a lower bound to the cardinality of the strongest nonlocal set $\mathcal{S}$ in $\otimes _{i=1}^{n}\mathbb{C}^{d_i}$, i.e., $|\mathcal{S}|\leq \max_{i}\{\prod_{j=1}^{n}d_j/d_i+1\}$. In this work, we construct the strongest nonlocal set of size $d^2+1$ in $\mathbb{C}^{d}\otimes \mathbb{C}^{d}\otimes \mathbb{C}^{d}$. Furthermore, we obtain the strongest nonlocal set of size $d_{2}d_{3}+1$ in $\mathbb{C}^{d_1}\otimes \mathbb{C}^{d_2}\otimes \mathbb{C}^{d_3}$. Our construction reaches the lower bound, which provides an affirmative solution to Li and Wang's conjecture. In particular, the strongest nonlocal sets we present here contain the least number of orthogonal states among the available results.
		\begin{description}
			\item[PACS numbers]
			03.65.Ud, 03.67.Mn
		\end{description}
	\end{abstract}
	
	\pacs{Valid PACS appear here}
	\maketitle
	
	
	\section{\label{sec:level1}Introduction\protect}

In 1964, Bell derived the famous Bell's inequality \cite{Bell}. Since then, many experiments have been performed to demonstrate violation of Bell's inequality, indicating that the pure entangled states have Bell nonlocality. Unlike Bell nonlocality, in 1999, Bennett {\it et al}. \cite{Bennett1999} found that the set of orthogonal product states, which is not perfectly distinguishable under local operations and classical communication (LOCC), also reflects quantum nonlocality. In Ref. \cite{Mal}, Mal and Sen attempted to unify the concepts of Bell nonlocality and Bennett nonlocality in quantum information theory.

A known set of orthogonal quantum states is said to have the property of quantum nonlocality if it is not possible to distinguish them via LOCC. Quantum nonlocality, i.e., Bennett nonlocality, which is widely used in quantum secret sharing and quantum data hiding \cite{Yang2015,Wangj2017,Jiangdh2020,GuoGP2001,Rahaman2015,DiVincenzo2002}, has given rise to research and fruitful results in the past two decades \cite{Walgate2000,Ghosh2001,Walgate2002,Ghosh2004,Chen2004,Niset2006, Cohen2007,Xin2008,Zhang2014,Yu2015,Wang2015,Zhang2015,Zhangxq2016,Xu2016,Zhangzc2016,Wang2017,Zhang2017,Halder2018,Jiang2020,Xu2021, Zuo2022,Zhen2022,Zhu2023,Cao2023}.

	A stronger manifestation of quantum nonlocality  in multipartite quantum systems, known as strong quantum nonlocality, was discovered by Halder  {\it et al.}  \cite{Halder2019}. A set of orthogonal states is said to be strongly nonlocal if it is locally irreducible in every bipartition. Indeed,  a locally irreducible set is locally indistinguishable, but the  converse is not true in general. They also constructed orthogonal product bases with strong nonlocality in $\mathbb{C}^3\otimes\mathbb{C}^3\otimes\mathbb{C}^3$ and $\mathbb{C}^4\otimes\mathbb{C}^4\otimes\mathbb{C}^4$. Soon after, numerous results have emerged concerning the existence of orthogonal product sets (OPSs) and orthogonal entangled sets (OESs) with strong nonlocality   \cite{Halder2019,Zhangzc2019,Yuan2020,Zhou2022,He,Zhou2023,Shi2022Quantum, Shi2021JPA,Che2022,Shi2020PRA,Wang2021,Shi2022PRA,Hu2024,Li2023,Li2023PRA} (see Table~\ref{results} for a summary).
	
	For strongly nonlocal OPSs, Yuan {\it et al}. \cite{Yuan2020} presented a strongly nonlocal OPS of $6(d-1)^2$ (with respect to $6d^2-8d+4$) elements in $\mathbb{C}^d\otimes\mathbb{C}^d\otimes\mathbb{C}^d$ (with respect to $\mathbb{C}^d\otimes\mathbb{C}^d\otimes\mathbb{C}^{d+1}$). Meanwhile, they gave some examples of strongly nonlocal OPSs in $\mathbb{C}^3\otimes\mathbb{C}^3\otimes\mathbb{C}^3\otimes\mathbb{C}^3$ and $\mathbb{C}^4\otimes\mathbb{C}^4\otimes\mathbb{C}^4\otimes\mathbb{C}^4$. After that, Zhou {\it et al}. proposed strongly nonlocal OPSs of smaller size in arbitrary three and four-partite systems \cite{Zhou2022}. By using a general decomposition of the $N$-dimen- sional hypercubes for odd $N\geq 3$, He {\it et al}. \cite{He} showed a strongly nonlocal OPS of size $\prod_{i=1}^{N} d_{i}-\prod_{i=1}^{N} (d_{i}-2)$ in $\otimes _{i=1}^{N}\mathbb{C}^{d_i}$. In general $n$-partite systems with even $n$, Zhou {\it et al}. put forward a construction of strongly nonlocal OPSs \cite{Zhou2023}. In particular, an unextendible product basis (UPB) can also exhibit strong nonlocality. In Ref. \cite{Shi2022Quantum}, Shi {\it et al}. showed the existence of UPBs that are locally irreducible in every bipartition in $\mathbb{C}^d\otimes\mathbb{C}^d\otimes\mathbb{C}^d$. Moreover, the above construction can be generalized to arbitrary three- and four-partite systems \cite{Shi2021JPA}. In 2022, Che {\it et al}. \cite{Che2022} provided a strongly nonlocal UPB with cardinality $(d-1)^3+2d+5$ in $\mathbb{C}^d\otimes\mathbb{C}^d\otimes\mathbb{C}^d$ and generalized this approach to arbitrary tripartite systems.
	
	For strongly nonlocal OESs, Shi {\it et al}. \cite{Shi2020PRA} obtained the first result of strongly nonlocal sets with entanglement. In Ref. \cite{Wang2021}, Wang {\it et al}. related orthogonal genuinely entangled sets (OGESs) with strong nonlocality to the connectivities of graphs. They also proposed strongly nonlocal OGESs of size $d^3-(d-2)^2$ ($d$ is odd) and $d^3-(d-2)^2+2$ ($d$ is even) in $\mathbb{C}^d\otimes\mathbb{C}^d\otimes\mathbb{C}^d$. Further, they extended this result to the general case. Several strongly nonlocal OESs with cardinality $d^N-(d-1)^N+1$ have been presented by Shi {\it et al}. in $(\mathbb{C}^d)^{\otimes N}$, where $N\geq 3$ and $d\geq 2$ \cite{Shi2022PRA}. When $N=3$ and $4$, the sets they constructed were strongly nonlocal OGESs. In $N$-qutrit systems, Hu {\it et al}. \cite{Hu2024} constructed a strongly nonlocal OGES of size $2\times 3^{N-1}$.
		\begin{table*}[htbp]
		\newcommand{\tabincell}[2]{\begin{tabular}{@{}#1@{}}#2\end{tabular}}
		\centering
		\caption{\label{results}Results about the strongest nonlocal sets in $\mathbb{C}^{d}\otimes \mathbb{C}^{d}\otimes \mathbb{C}^{d}$ or $\mathbb{C}^{d_1}\otimes \mathbb{C}^{d_2}\otimes \mathbb{C}^{d_3}$. }
		\begin{tabular}{ccc}
			\toprule
			\hline
			\hline
			\specialrule{0em}{1.5pt}{1.5pt}
			~Type~~&~~Cardinality~~&~~~References~~\\
			\specialrule{0em}{1.5pt}{1.5pt}
			\midrule
			\hline
			\specialrule{0em}{1.5pt}{1.5pt}
			\tabincell{c}{OPS}&\tabincell{c}{$6(d-1)^2$}&~~\tabincell{c}{\cite{Yuan2020}} \\
			\specialrule{0em}{3.5pt}{3.5pt}
			\tabincell{c}{OPS}&\tabincell{c}{$2[(d_{1}d_{2}+d_{2}d_{3}+d_{1}d_{3})$\\$-3(d_1+d_2+d_3)+12]$}&~~\tabincell{c}{\cite{Zhou2022}} \\
			\specialrule{0em}{3.5pt}{3.5pt}
			\tabincell{c}{UPB}&\tabincell{c}{$d_{1}d_{2}d_{3}-8(n+1),~0\leq n\leq \lfloor \frac{d_{1}-3}{2}\rfloor$}&~~\tabincell{c}{\cite{Shi2021JPA}}\\
			\specialrule{0em}{3.5pt}{3.5pt}
			\tabincell{c}{UPB}&\tabincell{c}{$(d-1)^3+2d+5$}&~~\tabincell{c}{\cite{Che2022}}\\
			\specialrule{0em}{3.5pt}{3.5pt}
			\tabincell{c}{OES}&\tabincell{c}{$d^3-d$ ($d$ is odd)\\$d^3-d-6$ ($d$ is even)}&~~\tabincell{c}{\cite{Shi2020PRA}}\\
			\specialrule{0em}{3.5pt}{3.5pt}
			\tabincell{c}{OGES}&\tabincell{c}{$d^3-(d-2)^2$ ($d$ is odd)\\$d^3-(d-2)^2+2$ ($d$ is even)}&~~\tabincell{c}{\cite{Wang2021}}\\
			\specialrule{0em}{3.5pt}{3.5pt}
\tabincell{c}{OGES}&\tabincell{c}{$d^3-(d-1)^3+1$}&~~\tabincell{c}{\cite{Shi2022PRA}}\\
			\specialrule{0em}{3.5pt}{3.5pt}		
			\tabincell{c}{OGES\\and $|000\rangle$}&\tabincell{c}{$\prod_{i=1}^{3}d_{i}-\prod_{i=1}^{3}(d_i-1)$}&~~\tabincell{c}{\cite{Li2023}}\\
			\specialrule{0em}{3.5pt}{3.5pt}
			\tabincell{c}{OGES\\and $|000\rangle$}&\tabincell{c}{$d_2d_3+d_1-1$}&~~\tabincell{c}{\cite{Li2023PRA}}\\		
			\specialrule{0em}{3.5pt}{3.5pt}
			\tabincell{c}{OGES \\and $|S_2\rangle$}&\tabincell{c}{$d^2+1$~(the lower bound)}&~~\tabincell{c}{Theorem~\ref{th:ddd}}\\
			\specialrule{0em}{3.5pt}{3.5pt}
			\tabincell{c}{OES \\and $|S\rangle$}&\tabincell{c}{$d_{2}d_{3}+1$~(the lower bound)}&~~\tabincell{c}{Theorem~\ref{th:d1d2d3}}\\		
			\bottomrule
			\specialrule{0em}{1.5pt}{1.5pt}
			\hline
			\hline
		\end{tabular}
	\end{table*}

	Recently, Li and Wang \cite{Li2023} proposed the definition of a locally stable set and gave two conjectures: (a) There exists the smallest set of cardinality max$_{i}\{d_i+1\}$ of orthogonal states that is locally stable in $\otimes _{i=1}^{n}\mathbb{C}^{d_i}$. (b) The smallest set exhibiting the strongest nonlocality in $\otimes_ {i=1}^{n}\mathbb{C}^{d_i}$ has a cardinality of $ \max_ {i}\{\prod_{j=1}^{n}d_j/d_i+1\}$. In Ref. \cite{Cao2023PRA}, Cao {\it et al}. provided a locally stable set of size max$_{i}\{d_i+1\}$, which showed that the first conjecture holds. Li {\it et al}. \cite{Li2023PRA} constructed the strongest nonlocal sets of size $d_2d_3+d_1-1$ (with respect to $d^3+d-1$) in $\mathbb{C}^{d_1}\otimes\mathbb{C}^{d_2}\otimes\mathbb{C}^{d_3}$ (with respect to $\mathbb{C}^{d}\otimes\mathbb{C}^{d}\otimes\mathbb{C}^{d}\otimes\mathbb{C}^{d}$). The above constructions reach the lower bound only in $\mathbb{C}^2\otimes\mathbb{C}^{d_2}\otimes\mathbb{C}^{d_3}$ and $\mathbb{C}^{2}\otimes\mathbb{C}^{2}\otimes\mathbb{C}^{2}\otimes\mathbb{C}^{2}$. Up to now, whether the conjecture holds in general tripartite systems remains an open question worth considering.
	
	In this work, we construct the strongest nonlocal sets with minimum cardinality in general tripartite systems, confirming the conjecture raised in Ref. \cite{Li2023}. First, we present the strongest nonlocal set of size $d^2+1$ in $\mathbb{C}^d\otimes\mathbb{C}^d\otimes\mathbb{C}^d$ for $d\geq 3$, which consists of the stopper state and special genuinely entangled states [Greenberger-Horne-Zeilinger](GHZ)-like and $W$-like states, the corresponding Rubik's cube exhibits the perfect geometric symmetry. Moreover, we prove that the smallest orthogonal set of size $21$ is the strongest nonlocal set in $\mathbb{C}^3\otimes\mathbb{C}^4\otimes\mathbb{C}^5$ and provides the structure to general tripartite systems $\mathbb{C}^{d_1}\otimes\mathbb{C}^{d_2}\otimes\mathbb{C}^{d_3}$.
	
	\theoremstyle{remark}
	\newtheorem{definition}{\indent Definition}
	\newtheorem{lemma}{\indent Lemma}
	\newtheorem{theorem}{\indent Theorem}
	\newtheorem{proposition}{\indent Proposition}
	\newtheorem{observation}{\indent Observation}
	\newtheorem{example}{\indent Example}
	\newtheorem{corollary}{\indent Corollary}
	\def\QEDclosed{\mbox{\rule[0pt]{1.3ex}{1.3ex}}}
	\def\QED{\QEDclosed}
	\def\proof{\indent{\em Proof}.}
	\def\endproof{\hspace*{\fill}~\QED\par\endtrivlist\unskip}

\section{Preliminaries}
Throughout this paper, we only consider pure quantum states. For the sake of convenience, the states discussed in this paper are un-normalized quantum states. We define $\mathbb Z_{d}=\{0, 1,\cdots, d-1\}$ and choose the computational basis $\{|i\rangle \mid i\in \mathbb Z_{d_k}\}$ for each $d_k$-dimensional subsystem. For each integer $d\geq 2$, $\omega_d=e^{\frac{2\pi \sqrt{-1}}{d}}$, i.e., a primitive $d$-th root of unit.

Consider a positive operator-valued measure (POVM) that performed on subsystems, each POVM element can be represented by a $d\times d$ matrix, denote $E=(m_{i,j})_{i,j\in \mathbb Z_{d}}$, under the computational basis of subsystems.  A measurement is an \emph{trivial} if all its POVM elements are proportional to the identity operator. Otherwise, the measurement is called nontrivial. A measurement is \emph{orthogonality-preserving local measurement} (OPLM) if the post-measurement states are mutually orthogonal.

\begin{definition}({\bf The strongest nonlocality})\ucite{Shi2022Quantum}
A set $\mathcal{S}$ of orthogonal multipartite  quantum  states  is said to be of \emph{the strongest nonlocality} if only trivial OPLM can be performed for each bipartition of the subsystems.
\end{definition}

In a general tripartite quantum system $ \mathcal{H}_A\otimes \mathcal{H}_B\otimes \mathcal{H}_C=\mathbb{C}^{d_1}\otimes\mathbb{C}^{d_2} \otimes \mathbb{C}^{d_3}$ ($d_1\leq d_2\leq d_3$), there are three different bipartitions: $A|BC$, $B|CA$, and $C|AB$.  Given a set $\mathcal{S}$ of orthogonal states in $ \mathcal{H}_A\otimes \mathcal{H}_B\otimes \mathcal{H}_C$, to demonstrate the set $\mathcal{S}$  is the strongest nonlocal set, it is sufficient to show that each joint party $BC$, $CA$, and $AB$ can only start with a trivial OPLM. For example,  if $E_{BC}=M_{BC}^{\dagger}M_{BC}$ is one of $BC$'s measurement elements, we use the orthogonality relations $\{\mathbb{I}_A\otimes M_{BC} |\Psi\rangle \}_{|\Psi\rangle \in \mathcal{S}}$, that is,
\begin{equation}\label{eq:OR}
	\langle \Phi|\mathbb{I}_A\otimes   E_{BC} |\Psi\rangle= \langle \Phi|\mathbb{I}_A\otimes M_{BC}^\dagger M_{BC} |\Psi\rangle =0,
\end{equation}
$\forall ~|\Phi\rangle \neq |\Psi\rangle \in \mathcal{S},$ to show that $E_{BC}\propto \mathbb{I}_{BC}$ ($E_{CA}$ and $E_{AB}$ are similar). Throughout this paper, we assume $E_{BC}=(m_{ij,kl})_{ij,kl\in \mathbb{Z}_{d_2} \times \mathbb{Z}_{d_3}}$, $E_{CA}=(m_{ij,kl})_{ij, kl \in \mathbb Z_{d_3}\times \mathbb Z_{d_1}}$ and $E_{AB}=(m_{ij,kl})_{ij, kl \in \mathbb Z_{d_1}\times \mathbb Z_{d_2}}$.

 A state $|\phi\rangle_{ABC}$ is called \emph{a genuinely entangled state} if it is entangled in every bipartition. It is well known that the three-qubit GHZ state ($|000\rangle+|111\rangle)/\sqrt{2}$ and the $W$ state ($|100\rangle+|010\rangle+|001\rangle)/\sqrt{3}$ are genuinely entangled states.
The states considered throughout this paper (see Fig. \ref{fig:333} for an example)  can be separated into the following three classes:

\begin{enumerate}
	\item [\rm(i)]\emph{the GHZ-like states:} $|i_1\rangle_{A}|j_1\rangle_{B}|k_1\rangle_{C}- |i_2\rangle_{A}|j_2\rangle_{B}$ $|k_2\rangle_{C}$, where  $i_1\neq i_2$, $j_1\neq j_2$ and $k_1\neq k_2$;
		\item [\rm(ii)]\emph{the W-like states:}  $|i\rangle_{A}|j\rangle_{B}|k\rangle_{C}+\omega_{3}^{1}|j\rangle_{A}|k\rangle_{B}|i\rangle_{C}+\omega_{3}^{2}|k\rangle_{A}|i\rangle_{B}|j\rangle_{C}$, where $i, j$, and $k$ are not all the same, i.e., $(i,j,k)$ is not of the form $(a,a,a)$; and
	\item [\rm(iii)]\emph{the stopper state:}\\$|S\rangle= \left(\sum\limits_{i\in \mathbb Z_{d_1}}|i\rangle_{A}\right)\left(\sum\limits_{j\in \mathbb Z_{d_2}}|j\rangle_{B}\right)\left(\sum\limits_{k\in \mathbb Z_{d_3}}|k\rangle_{C}\right)$.
\end{enumerate}

 For a GHZ-like state,  $|\Phi\rangle=|i_1\rangle_{A}|j_1\rangle_{B}|k_1\rangle_{C}- |i_2\rangle_{A}|j_2\rangle_{B}$ $|k_2\rangle_{C}$,  we refer to $(i_1, j_{1}k_1)$ and $(i_2, j_{2}k_2)$ as two cells of $|\Phi\rangle$ with respect to the bipartition $A|BC$, where $i_1$ and $i_2$ are the row indexes, respectively.  For a $W$-like state,  $|\Phi\rangle=|i\rangle_{A}$ $|j\rangle_{B}|k\rangle_{C}+\omega_{3}^{1}|j\rangle_{A}|k\rangle_{B}|i\rangle_{C}+\omega_{3}^{2}|k\rangle_{A}|i\rangle_{B}|j\rangle_{C}$, we refer to $(i, jk),$  $(j, ki)$, and $(k,ij)$ as three cells of $|\Phi\rangle$ with respect to the bipartition $A|BC$,  where   $i,j$, and $k$ are the row indexes, respectively.   Building on the orthogonality of the post-measurement states, we present the following observations that may be significant in demonstrating that the joint party $BC$ can only start with trivial measurements ($CA$ and $AB$ can be dealt with similarly).

\begin{observation}\label{ob1}
Let $|\Phi\rangle$ and $|\Psi\rangle$ be orthogonal GHZ-like states or $W$-like states. If only one pair of cells  of $|\Phi\rangle$ and $|\Psi\rangle$ with respect to the bipartition $A|BC$ have the same row index, denoted as $(i, jk)$ and $(i, j'k')$, then applying $|\Phi\rangle$ and $|\Psi\rangle$ to Eq. \eqref{eq:OR} results in:
   $$m_{jk,j'k'}=0.$$
\end{observation}

For example, consider two orthogonal GHZ-like states $|\phi_{20}\rangle=|2\rangle_{A}|2\rangle_{B}|0\rangle_{C}-|1\rangle_{A}|0\rangle_{B}|2\rangle_{C}$ and $|\phi_{02}\rangle=|2\rangle_{A}|0\rangle_{B}$ $|2\rangle_{C}-|0\rangle_{A}|2\rangle_{B}|1\rangle_{C}$. In Fig. \ref{fig:333A}, we find that only two cells $(2,20)$ and $(2,02)$ have the unique same row index ``2'' in the bipartition $A|BC$. Applying Obervation \ref{ob1} to $|\phi_{20}\rangle$ and $|\phi_{02}\rangle$, we have
the equation, that is, $(_{A}\langle2|_{BC}\langle20|-_{A}\langle1|
_{BC}\langle02|)\mathbb{I}_A\otimes M_{BC}^{\dagger}M_{BC}(|2\rangle_{A}|02\rangle_{BC}-|0\rangle_{A}|21\rangle_{BC})=0$. Hence, we get that $m_{20,02}=0$. Since $E=E^{\dag}$, we have $m_{02,20}=m_{20,02}=0$.

\begin{observation}\label{ob2}
	
Let $|\Phi\rangle$ and $|\Psi\rangle$ be orthogonal  GHZ-like states or $W$-like states. If more than one pair of cells  of $|\Phi\rangle$ and $|\Psi\rangle$ with respect to the bipartition $A|BC$ have the same row index, denoted as $(i, jk)$ and $(i, j'k')$, and as $(i_x, j_xk_x)$ and $(i_x, j'_xk'_x)$ with $x\in \mathcal{X}$ for some index set $\mathcal{X}$, then applying $|\Phi\rangle$ and $|\Psi\rangle$ to Eq. \eqref{eq:OR} results in:
$$a m_{jk,j'k'}+\sum_{x\in\mathcal{X}} a_xm_{j_xk_x,j_x'k_x'}  =0,$$
where $|a|=|a_x|=1.$ If we have known  $m_{j_xk_x,j_x'k_x'}=0$ for $x\in \mathcal{X}$, then we have $$m_{jk,j'k'}=0.$$

\end{observation}
The cardinality of the index set $\mathcal{X}$ is either $1$ or $2$ in this paper.
For example, given $|\phi_{20}\rangle$ and $|\phi_{12}\rangle=|2\rangle_{A}|1\rangle_{B}$ $|2\rangle_{C}-|1\rangle_{A}|2\rangle_{B}|0\rangle_{C}$, we cannot use Obervation \ref{ob1} directly. However, we can obtain the equation by Observation \ref{ob2}, i.e., $(_{A}\langle2|_{BC}\langle20|-_{A}\langle1|_{BC}\langle02|)\mathbb{I}_A\otimes M_{BC}^{\dagger}M_{BC}$ $(|2\rangle_{A}|12\rangle_{BC}-|1\rangle_{A}|20\rangle_{BC})=0$. It implies that $m_{20,12}+m_{02,20}=0$. As $m_{02,20}=0$ can be yielded from Observation \ref{ob1}; therefore $m_{20,12}=0$.

\begin{observation}\label{ob3}
	Suppose that all off-diagonal entries in the matrix $E_{BC}$ are zeros.
Applying  the stopper state $|S\rangle$ and a GHZ-like state $|\Phi\rangle
=|a\rangle_{A}|i\rangle_{B}|j\rangle_{C}-|b\rangle_{A}|k\rangle_{B}|l\rangle_{C}$ to Eq. \eqref{eq:OR} yields
$$m_{ij,ij}= m_{kl,kl}. $$
\end{observation}

\begin{observation}\label{ob4}
Suppose that all off-diagonal entries in the matrix $E_{BC}$ are zeros.
Applying  the stopper state $|S\rangle$ and the $W$-like state $|\Psi\rangle
=|i\rangle_{A}|j\rangle_{B}|k\rangle_{C}+\omega_{3}^{1}|j\rangle_{A}|k\rangle_{B}|i\rangle_{C}+\omega_{3}^{2}|k\rangle_{A}|i\rangle_{B}|j\rangle_{C}$ to Eq. \eqref{eq:OR} yields
$$m_{ij,ij}= m_{jk,jk}=m_{ki,ki}.
$$
\end{observation}

In fact, under the condition that all off-diagonal entries in the matrix $E_{BC}$ are zeros, Eq. \eqref{eq:OR} is just
\begin{equation}\label{eq:identity1}
m_{jk,jk}+\omega_{3}^{1}m_{ki,ki}+\omega_{3}^{2}m_{ij,ij}=0,
\end{equation}
where $m_{jk,jk},m_{ki,ki}$, and $m_{ij,ij}$ are all real numbers. Taking complex conjugation to both sides of Eq. \eqref{eq:identity1}, we have
\begin{equation}\label{eq:identity2}
	m_{jk,jk}+\omega_{3}^{2}m_{ki,ki}+\omega_{3}^{1}m_{ij,ij}=0.
\end{equation}
Solving the two linear equations, Eqs.~\eqref{eq:identity1} and \eqref{eq:identity2}, one has
$(m_{jk,jk},m_{ki,ki},m_{ij,ij})=(r,r,r)$
for some real number $r$.

\section{The strongest nonlocal sets with minimum cardinality in $\mathbb{C}^{d}\otimes \mathbb{C}^{d}\otimes \mathbb{C}^{d}$}

In this section, we present an orthogonal set of size $10$ with the strongest nonlocality in $\mathbb{C}^{3}\otimes \mathbb{C}^{3}\otimes \mathbb{C}^{3}$, the corresponding Rubik's cube of which is given by Fig. \ref{fig:333}. Subsequently, we propose the strongest nonlocal set with minimum cardinality in $\mathbb{C}^{d}\otimes \mathbb{C}^{d}\otimes \mathbb{C}^{d}$, where $d\geq 3$.

\begin{figure}[ht]
	\centering
	\includegraphics[width=4.5cm]{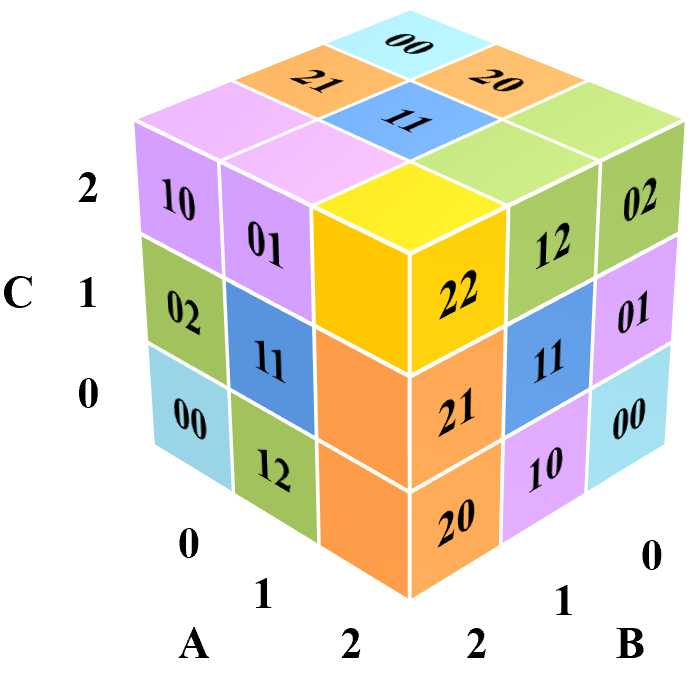}
	\caption{The Rubik's cube corresponding to the strongest nonlocal set with minimum cardinality given by Eq. \eqref{eq:333} in $\mathbb{C}^{3}\otimes\mathbb{C}^{3}\otimes \mathbb{C}^{3}$. Every cube has an index $(i,j,k)$, where $i,j$, and $k$ are the indexes in each subsystem $A, B$, and $C$. For example, the GHZ-like state $|\phi_{22}\rangle=|2\rangle_{A}|2\rangle_{B}|2\rangle_{C}-|1\rangle_{A}|1\rangle_{B}|1\rangle_{C}$ corresponds to two yellow cubes $(1,1,1)$ and $(2,2,2)$ with label ``22''. The $W$-like state $|\phi_{00}\rangle=|2\rangle_{A}|0\rangle_{B}|0\rangle_{C}+\omega_{3}^{1}|0\rangle_{A}$ $|0\rangle_{B}|2\rangle_{C}+\omega_{3}^{2}|0\rangle_{A}|2\rangle_{B}|0\rangle_{C}$ corresponds to three blue cubes $(2,0,0),(0,0,2)$, and $(0,2,0)$ with label ``00''.} \label{fig:333}
\end{figure}

\begin{table*}[htbp]
	\newcommand{\tabincell}[2]{\begin{tabular}{@{}#1@{}}#2\end{tabular}}
	\centering
	\caption{\label{ta1}The off-diagonal entries in the matrix $E_{BC}=(m_{ij,kl})_{ij\neq kl\in \mathbb Z_{3}\times \mathbb Z_{3}}$ by Observations~\ref{ob1} and \ref{ob2}.}
	\begin{tabular}{ccccccc}
		\toprule
		\hline
		\hline
		\specialrule{0em}{1.5pt}{1.5pt}
		Observations~&~Pair of states&Key entries~&Pair of states&Remaining entries~&Pair of states&Remaining entries\\
		\specialrule{0em}{1.5pt}{1.5pt}
		\midrule
		\hline
		\specialrule{0em}{1.5pt}{1.5pt}
		\tabincell{c}{Observation~\ref{ob1}}
		&\tabincell{c}{$|\phi_{02}\rangle, |\phi_{11}\rangle$\\$|\phi_{02}\rangle, |\phi_{12}\rangle$\\$|\phi_{02}\rangle, |\phi_{20}\rangle$\\$|\phi_{02}\rangle, |\phi_{22}\rangle$\\$|\phi_{11}\rangle, |\phi_{21}\rangle$\\$|\phi_{12}\rangle, |\phi_{21}\rangle$\\$|\phi_{20}\rangle, |\phi_{21}\rangle$\\$|\phi_{21}\rangle, |\phi_{22}\rangle$}&
		\tabincell{c}{$m_{02,11}=m_{11,02}=0$~~\\ $m_{02,12}=m_{12,02}=0$~~\\ $m_{02,20}=m_{20,02}=0$~~\\ $m_{02,22}=m_{22,02}=0$~~\\$m_{11,21}=m_{21,11}=0$~~\\ $m_{12,21}=m_{21,12}=0$~~\\ $m_{20,21}=m_{21,20}=0$\\$m_{21,22}=m_{22,21}=0$~~}
		&\tabincell{c}{$|\phi_{00}\rangle, |\phi_{11}\rangle$\\$|\phi_{00}\rangle, |\phi_{01}\rangle$\\$|\phi_{00}\rangle, |\phi_{12}\rangle$\\$|\phi_{00}\rangle, |\phi_{20}\rangle$\\$|\phi_{00}\rangle, |\phi_{22}\rangle$\\$|\phi_{11}\rangle, |\phi_{10}\rangle$}&
		\tabincell{c}{$m_{00,11}=m_{11,00}=0$~~\\ $m_{00,01}=m_{01,00}=0$~~\\ $m_{00,12}=m_{12,00}=0$~~\\ $m_{00,20}=m_{20,00}=0$~~\\ $m_{00,22}=m_{22,00}=0$~~\\$m_{11,10}=m_{10,11}=0$~~}&
		\tabincell{c}{$|\phi_{01}\rangle, |\phi_{02}\rangle$\\$|\phi_{01}\rangle, |\phi_{10}\rangle$\\$|\phi_{01}\rangle, |\phi_{21}\rangle$\\$|\phi_{10}\rangle, |\phi_{12}\rangle$\\$|\phi_{10}\rangle, |\phi_{20}\rangle$\\$|\phi_{10}\rangle, |\phi_{22}\rangle$}&
		\tabincell{c}{$m_{01,02}=m_{02,01}=0$~~\\$m_{01,10}=m_{10,01}=0$~~\\ $m_{01,21}=m_{21,01}=0$~~\\ $m_{10,12}=m_{12,10}=0$~~\\$m_{10,20}=m_{20,10}=0$~~\\ $m_{10,22}=m_{22,10}=0$~~}\\
		\specialrule{0em}{3.5pt}{3.5pt}
		\tabincell{c}{Observation~\ref{ob2}}
		&\tabincell{c}{$|\phi_{02}\rangle, |\phi_{21}\rangle$\\$|\phi_{12}\rangle, |\phi_{20}\rangle$\\ $|\phi_{20}\rangle, |\phi_{22}\rangle$\\$|\phi_{11}\rangle, |\phi_{12}\rangle$\\$|\phi_{11}\rangle, |\phi_{20}\rangle$\\$|\phi_{11}\rangle, |\phi_{22}\rangle$\\$|\phi_{12}\rangle, |\phi_{22}\rangle$}
		&\tabincell{c}{$m_{02,21}=m_{21,02}=0$~~\\ $m_{12,20}=m_{20,12}=0$~~\\ $m_{20,22}=m_{22,20}=0$~~\\ $m_{11,12}=m_{12,11}=0$~~\\$m_{11,20}=m_{20,11}=0$~~\\$m_{11,22}=m_{22,11}=0$~~\\$m_{12,22}=m_{22,12}=0$~~}&
		\tabincell{c}{$|\phi_{01}\rangle, |\phi_{12}\rangle$\\$|\phi_{01}\rangle, |\phi_{20}\rangle$\\ $|\phi_{01}\rangle, |\phi_{22}\rangle$\\$|\phi_{10}\rangle, |\phi_{02}\rangle$\\$|\phi_{10}\rangle, |\phi_{21}\rangle$}
		&\tabincell{c}{$m_{01,12}=m_{12,01}=0$~~\\$m_{01,20}=m_{20,01}=0$~~\\ $m_{01,22}=m_{22,01}=0$~~\\$m_{10,02}=m_{02,10}=0$~~\\$m_{10,21}=m_{21,10}=0$~~}
		&\tabincell{c}{$|\phi_{00}\rangle, |\phi_{02}\rangle$\\$|\phi_{00}\rangle, |\phi_{10}\rangle$\\$|\phi_{00}\rangle, |\phi_{21}\rangle$\\$|\phi_{11}\rangle, |\phi_{01}\rangle$}
		&\tabincell{c}{$m_{00,02}=m_{02,00}=0$~~\\$m_{00,10}=m_{10,00}=0$~~\\ $m_{00,21}=m_{21,00}=0$~~\\ $m_{11,01}=m_{01,11}=0$~~}\\
		\bottomrule
		\specialrule{0em}{1.5pt}{1.5pt}
		\hline
		\hline
	\end{tabular}
\end{table*}

\begin{lemma}\label{le:333}    The set $\cup_{i,j\in\mathbb{Z}_3}\{|\phi_{ij}\rangle\}\bigcup \{|S_1\rangle\}$ of size $10$ given by Eq. \eqref{eq:333} is strongest nonlocal in $\mathbb{C}^{3}\otimes \mathbb{C}^{3}\otimes \mathbb{C}^{3}$:
\begin{equation}\label{eq:333}
\begin{aligned}
|\phi_{22}\rangle=&|2\rangle_{A}|2\rangle_{B}|2\rangle_{C}-|1\rangle_{A}|1\rangle_{B}|1\rangle_{C},\\
|\phi_{20}\rangle=&|2\rangle_{A}|2\rangle_{B}|0\rangle_{C}-|1\rangle_{A}|0\rangle_{B}|2\rangle_{C},\\
|\phi_{21}\rangle=&|2\rangle_{A}|2\rangle_{B}|1\rangle_{C}-|0\rangle_{A}|1\rangle_{B}|2\rangle_{C},\\
|\phi_{02}\rangle=&|2\rangle_{A}|0\rangle_{B}|2\rangle_{C}-|0\rangle_{A}|2\rangle_{B}|1\rangle_{C},\\
|\phi_{12}\rangle=&|2\rangle_{A}|1\rangle_{B}|2\rangle_{C}-|1\rangle_{A}|2\rangle_{B}|0\rangle_{C},\\
|\phi_{10}\rangle=&|2\rangle_{A}|1\rangle_{B}|0\rangle_{C}-|0\rangle_{A}|2\rangle_{B}|2\rangle_{C},\\
|\phi_{01}\rangle=&|2\rangle_{A}|0\rangle_{B}|1\rangle_{C}-|1\rangle_{A}|2\rangle_{B}|2\rangle_{C},\\
|\phi_{00}\rangle=&|2\rangle_{A}|0\rangle_{B}|0\rangle_{C}+\omega_{3}^{1}|0\rangle_{A}|0\rangle_{B}|2\rangle_{C}+\omega_{3}^{2}|0\rangle_{A}|2\rangle_{B}|0\rangle_{C},\\
|\phi_{11}\rangle=&|2\rangle_{A}|1\rangle_{B}|1\rangle_{C}+\omega_{3}^{1}|1\rangle_{A}|1\rangle_{B}|2\rangle_{C}+\omega_{3}^{2}|1\rangle_{A}|2\rangle_{B}|1\rangle_{C},\\
|S_1\rangle=&|0+1+2\rangle_{A}|0+1+2\rangle_{B}|0+1+2\rangle_{C}.\\
\end{aligned}
\end{equation}
\end{lemma}

\begin{figure}
	\centering
	\includegraphics[width=6cm]{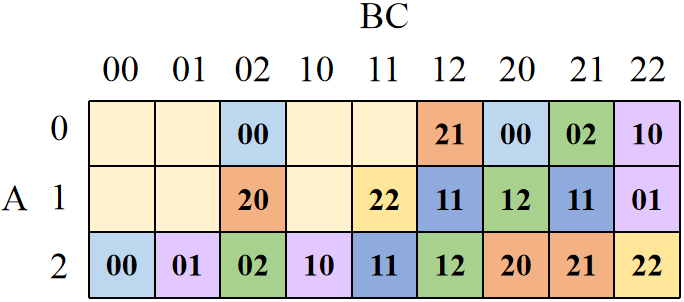}
	\vspace{0.5cm}
	\includegraphics[width=6cm]{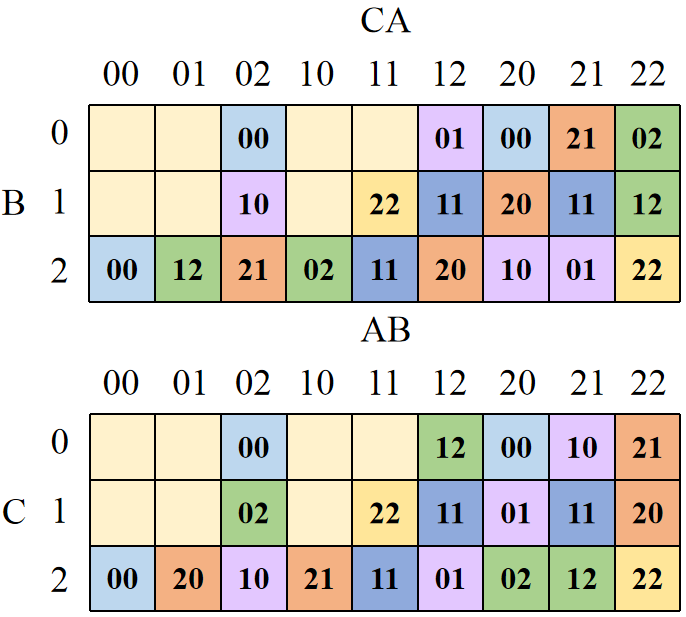}
	\caption{The $3\times 9$ plane structures of the strongest nonlocal set $\cup_{i,j\in\mathbb{Z}_3}\{|\phi_{ij}\rangle\}\bigcup \{|S_1\rangle\}$ correspond to each bipartition $A|BC$, $B|CA$, and $C|AB$. Every cell has an index $(i, jk)$, where $i$ is the row index in the single subsystem and $jk$ is the column index in the joint subsystem. For example, the GHZ-like state $|\phi_{20}\rangle=|2\rangle_{A}|20\rangle_{BC}-|1\rangle_{A}|02\rangle_{BC}$ corresponds to two cells $(2,20)$ and $(1,02)$ with label ``20'' in the bipartition $A|BC$. }  \label{fig:333A}
\end{figure}
\begin{proof}  Denote that $\mathcal S_{1}=\cup_{i,j\in\mathbb{Z}_3}\{|\phi_{ij}\rangle\}$. In Fig. \ref{fig:333A}, the set $\mathcal S_{1}\bigcup \{|S_{1}\rangle\}$ has a similar plane structure in every bipartition $A|BC$, $B|CA$, and $C|AB$ under the cyclic permutation. Thus, we only need to prove that the measurement applied to the joint subsystem $BC$ is trivial.

\begin{figure}[ht]
	\centering
	\includegraphics[width=8cm]{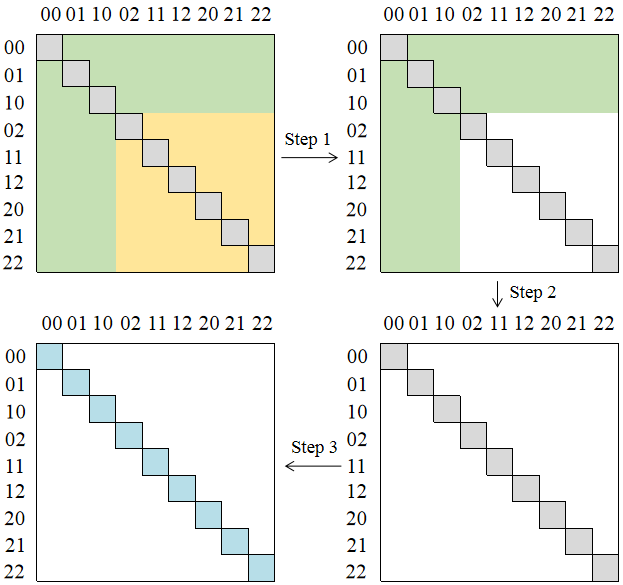}
	\caption{Applying observations to prove that the set $\cup_{i,j\in\mathbb{Z}_3}\{|\phi_{ij}\rangle\}\bigcup \{|S_1\rangle\}$ is the strongest nonlocal set. The yellow (green) entries correspond to the key (remaining) off-diagonal entries, the gray entries represent that all diagonal entries are unknown. Step $1$, we obtain that all key off-diagonal entries are zeros, which correspond to the white entries. Step $2$, all remaining off-diagonal entries are zeros. Step $3$, all diagonal entries are equal by Observations \ref{ob3} and \ref{ob4}, which are shown by the blue entries.}  \label{fig:333proof}
\end{figure}

\noindent{\bf Step $1$} According to Observations \ref{ob1} and \ref{ob2}, we obtain that all key off-diagonal entries are zeros, i.e., $m_{ij,kl}=0$ for $ij\neq kl\in \{02,$ $11,12,20,21,22\}$.

\noindent{\bf Step $2$} Consider $|\phi_{pq}\rangle$ and $|\phi_{st}\rangle$ for $pq\neq st$, $pq\in \{00,$ $01,10\}$, and $st\in \mathbb Z_{3}\times \mathbb Z_{3}$, we get that the remaining off-diagonal entries are zeros, that is, $m_{pq,st}=0$. Thus, all off-diagonal entries in the matrix $E_{BC}$ are zeros, that is, $m_{ij,kl}=m_{kl,ij}=0$ for $ij\neq kl\in \mathbb Z_{3}\times \mathbb Z_{3}$ from Table \ref{ta1}.
\begin{table*}[htbp]
	\newcommand{\tabincell}[2]{\begin{tabular}{@{}#1@{}}#2\end{tabular}}
	\centering
	\caption{\label{ta2}The key off-diagonal entries with respect to  coordinates of  $\mathcal S=\{\hat{d}i,j\hat{d},d^*d^*,\hat{d}\hat{d}\mid i,j\in \mathbb Z_{\hat{d}}\}$ in the matrix $E_{BC}=(m_{ij,kl})_{ij\neq kl\in \mathbb Z_{d}\times \mathbb Z_{d}}$ by Observations~\ref{ob1} and \ref{ob2}.}
	\begin{tabular}{cccccc}
		\toprule
		\hline
		\hline
		\specialrule{0em}{1.5pt}{1.5pt}
		~~~Observations~~~~&~~~~Sets~~~&~~~~Pair of states~~~&~~~Key entries~~~&~~~Value range~~~~&~~~Precondition~~~\\
		\specialrule{0em}{1.5pt}{1.5pt}
		\midrule
		\hline
		\specialrule{0em}{1.5pt}{1.5pt}
\tabincell{c}{Observation~\ref{ob1}}&\tabincell{c}{$\mathcal{A}_{1}$\\$\mathcal{A}_{2}$\\$\mathcal{A}_{1}, \mathcal{A}_{2}$\\$\mathcal{A}_{1}, \mathcal{A}_{5}$\\$\mathcal{A}_{2}, \mathcal{A}_{5}$\\$\mathcal{A}_{0}, \mathcal{A}_{1}$\\$\mathcal{A}_{0}, \mathcal{A}_{2}$}
&~~\tabincell{c}{$|\phi_{\hat{d}i}\rangle, |\phi_{\hat{d}j}\rangle$\\$|\phi_{i\hat{d}}\rangle, |\phi_{j\hat{d}}\rangle$\\$|\phi_{\hat{d}i}\rangle, |\phi_{j\hat{d}}\rangle$\\$|\phi_{\hat{d}i}\rangle, |\phi_{d^*d^*}\rangle$\\$| \phi_{i\hat{d}}\rangle, |\phi_{d^*d^*}\rangle$\\$|\phi_{\hat{d}\hat{d}}\rangle, |\phi_{\hat{d}i}\rangle$\\$|\phi_{\hat{d}\hat{d}}\rangle, |\phi_{i\hat{d}}\rangle$}
&~~\tabincell{c}{$m_{\hat{d}i,\hat{d}j}=m_{\hat{d}j,\hat{d}i}=0$~~\\ $m_{i\hat{d},j\hat{d}}=m_{j\hat{d},i\hat{d}}=0$~~\\ $m_{\hat{d}i,j\hat{d}}=m_{j\hat{d},\hat{d}i}=0$~~\\ $m_{\hat{d}i,d^*d^*}=m_{d^*d^*,\hat{d}i}=0$~~\\ $m_{i\hat{d},d^*d^*}=m_{d^*d^*,i\hat{d}}=0$~~\\ $m_{\hat{d}i,\hat{d}\hat{d}}=m_{\hat{d}\hat{d},\hat{d}i}=0$~~\\ $m_{i\hat{d},\hat{d}\hat{d}}=m_{\hat{d}\hat{d},i\hat{d}}=0$}

&~~\tabincell{c}{$i\neq j$~~\\ $i\neq j$~~\\ $i+j\neq d^*$~~\\ $ i\neq 0$~~\\$i\neq d^*$~~\\ $i\neq 0$~~\\ $i\neq  d^*$}

&~~\tabincell{c}{Null~~\\ Null~~\\ Null~~\\ Null~~\\Null~~\\ Null~~\\ Null }\\
	\specialrule{0em}{3.5pt}{3.5pt}
\tabincell{c}{Observation~\ref{ob2}}&\tabincell{c}{$\mathcal{A}_{1}, \mathcal{A}_{2}$\\$\mathcal{A}_{1}, \mathcal{A}_{5}$\\$\mathcal{A}_{2}, \mathcal{A}_{5}$\\$\mathcal{A}_{0}, \mathcal{A}_{1}$\\$\mathcal{A}_{0}, \mathcal{A}_{2}$\\$\mathcal{A}_{0}, \mathcal{A}_{5}$}
&~~\tabincell{c}{$|\phi_{i\hat{d}}\rangle, |\phi_{\hat{d}(d^*-i)}\rangle$\\$|\phi_{\hat{d}0}\rangle, |\phi_{d^*d^*}\rangle$\\$|\phi_{d^*\hat{d}}\rangle, |\phi_{d^*d^*}\rangle$\\$|\phi_{\hat{d}\hat{d}}\rangle, |\phi_{\hat{d}0}\rangle$\\$|\phi_{\hat{d}\hat{d}}\rangle, |\phi_{d^*\hat{d}}\rangle$\\$|\phi_{\hat{d}\hat{d}}\rangle, |\phi_{d^*d^*}\rangle$}
&~~\tabincell{c}{$m_{\hat{d}(d^*-i),i\hat{d}}=m_{i\hat{d},\hat{d}(d^*-i)}=0$~~\\ $m_{\hat{d}0,d^*d^* }=m_{d^*d^*,\hat{d}0}=0$~~\\ $m_{d^*\hat{d},d^*d^* }=m_{d^*d^*,d^*\hat{d}}=0$~~\\ $m_{\hat{d}0,\hat{d}\hat{d}}=m_{\hat{d}\hat{d},\hat{d}0}=0$~~\\ $m_{ d^* \hat{d},\hat{d}\hat{d}}=m_{\hat{d}\hat{d},d^*\hat{d}}=0$~~\\ $m_{d^*d^*,\hat{d}\hat{d}}=m_{\hat{d}\hat{d},d^*d^*}=0$}
&~~\tabincell{c}{$i \neq \frac{d^*}{2}$~~\\ $ $~~\\ $ $~~\\ $ $~~\\ $ $~~\\ $ $}
&~~\tabincell{c}{$m_{\hat{d}i,i\hat{d}}=0$~~\\  $m_{0\hat{d},d^*\hat{d}}=m_{0\hat{d},\hat{d}d^*}=0$~~\\ $m_{\hat{d}0,d^*\hat{d}}=m_{\hat{d}0,\hat{d}d^*}=0$~~\\ $m_{d^* d^*,0\hat{d}}=0$~~\\ $m_{\hat{d}0,d^* d^*}=0$~~\\ $m_{\hat{d}d^*,d^*d^*}=m_{d^*\hat{d},d^*d^*}=0$}\\
		\bottomrule
		\specialrule{0em}{1.5pt}{1.5pt}
		\hline
		\hline
	\end{tabular}
\end{table*}

\noindent{\bf Step $3$} Applying Observation \ref{ob3} to the stopper state $|S_{1}\rangle$ and $|\phi_{ij}\rangle$ for $ij\in \mathbb Z_3 \times \mathbb Z_3\backslash\{00,11\}$, we get $m_{02,02}=m_{12,12}=m_{20,20}=m_{21,21}$ and $m_{01,01}=m_{10,10}=m_{11,11}$ $=m_{22,22}$. Consider the stopper state $|S_{1}\rangle$ and the $W$-like states $|\phi_{00}\rangle$ and $|\phi_{11}\rangle$ by Observation \ref{ob4}, we have $m_{00,00}=m_{02,02}=m_{20,20}$ and $m_{11,11}= m_{12,12}=m_{21,21}$. In short, all diagonal entries are equal.

The proof that $E_{BC}$ is proportional to an identity matrix is shown in Fig. \ref{fig:333proof}. Therefore, the POVM element $E_{BC}$ is trivial.
This completes the proof.
\end{proof}

The latest result in $\mathbb{C}^{3}\otimes \mathbb{C}^{3}\otimes \mathbb{C}^{3}$ is the strongest nonlocal set of size $11$ given by Li {\it et al}. \cite{Li2023PRA}. Here, we provide the strongest nonlocal set of the smallest size $10$,
which positively answers an open problem raised by Yuan {\it et al}. \cite{Yuan2020} and reaches the lower bound on the strongest nonlocal set in $\mathbb{C}^{3}\otimes \mathbb{C}^{3}\otimes \mathbb{C}^{3}$ \cite{Li2023}.

Now, we put forward the strongest nonlocal set with minimum cardinality in three-qudit systems, the corresponding Rubik's cube of which is shown in Fig. \ref{fig:ddd}.

\begin{theorem} \label{th:ddd}  The set $\cup_{i=0}^{5}\mathcal{A}_{i}\bigcup \{|S_2\rangle\}$ of size $d^2+1$ given by Eq. \eqref{eq:ddd} is strongest nonlocal in $\mathbb{C}^d\otimes \mathbb{C}^d\otimes \mathbb{C}^d$ for $d\geq 3$:

\begin{equation}\label{eq:ddd}
\begin{aligned}
\mathcal{A}_{0}=&\{|\phi_{\hat{d}\hat{d}}\rangle=|\hat{d}\rangle_{A}|\hat{d}\rangle_{B}|\hat{d}\rangle_{C}-|d^*\rangle_{A}|d^*\rangle_{B}|d^*\rangle_{C}\},\\
\mathcal{A}_{1}=&\{|\phi_{\hat{d}i}\rangle=|\hat{d}\rangle_{A}|\hat{d}\rangle_{B}|i\rangle_{C}-|d^*-i\rangle_{A}|i\rangle_{B}|\hat{d}\rangle_{C}\big | i\in \mathbb{Z}_{\hat{d}}\},\\
\mathcal{A}_{2}=&\{|\phi_{i\hat{d}}\rangle=|\hat{d}\rangle_{A}|i\rangle_{B}|\hat{d}\rangle_{C}-|i\rangle_{A}|\hat{d}\rangle_{B}|d^*-i\rangle_{C}\big | i\in \mathbb{Z}_{\hat{d}}\},\\
\mathcal{A}_{3}=&\{|\phi_{(d^*-i)i }\rangle=|\hat{d}\rangle_{A}|d^*-i\rangle_{B}|i\rangle_{C}-|i\rangle_{A}|\hat{d}\rangle_{B}|\hat{d}\rangle_{C}  \mid i\in \mathbb{Z}_{\hat{d}}\},\\
\mathcal{A}_{4}=&\{|\phi_{kl}\rangle=
|\hat{d}\rangle_{A}|k\rangle_{B}|l\rangle_{C}+\omega_{3}^{1}|k\rangle_{A}|l\rangle_{B}|\hat{d}\rangle_{C}  \\
&+\omega_{3}^{2}|l\rangle_{A}|\hat{d}\rangle_{B}|k\rangle_{C}  \mid k,l\in \mathbb Z_{\hat{d}},k+l\geq  d-1  \},\\
\mathcal{A}_{5}=&\{|\phi_{st}\rangle=
 |\hat{d}\rangle_{A}|s\rangle_{B}|t\rangle_{C}+\omega_{3}^{1}|s\rangle_{A}|t\rangle_{B}|\hat{d}\rangle_{C}\\
 &+\omega_{3}^{2}|t\rangle_{A}|\hat{d}\rangle_{B}|s\rangle_{C} \mid  s,t\in \mathbb Z_{\hat{d}},0\leq s+t\leq d-3\},\\
 |S_2\rangle=& \left(\sum_{i\in \mathbb Z_{d}}|i\rangle_{A}\right)\left(\sum_{j\in \mathbb Z_{d}}|j\rangle_{B}\right)\left(\sum_{k\in \mathbb Z_{d}}|k\rangle_{C}\right),
\end{aligned}
\end{equation}
where $\hat{d}=d-1$ and $d^*=d-2.$
\end{theorem}
\begin{figure}[ht]
	\centering
	\includegraphics[width=6cm]{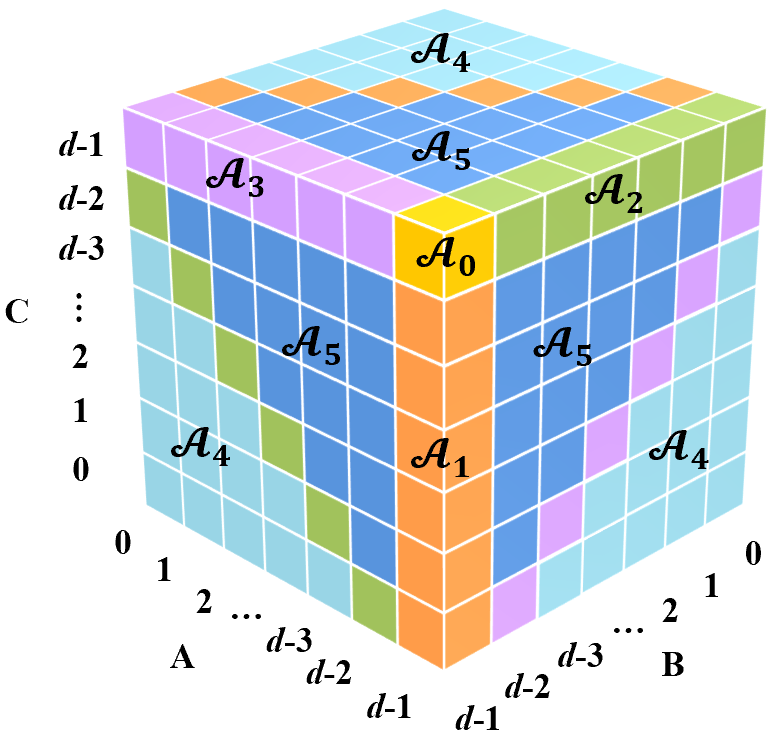}
	\caption{The Rubik's cube corresponding to the strongest nonlocal set with minimum cardinality given by Eq. \eqref{eq:ddd} in $\mathbb{C}^d\otimes$ $\mathbb{C}^d\otimes \mathbb{C}^d$ for $d\geq 3$. }  \label{fig:ddd}
\end{figure}
\begin{proof} The set $\cup_{i=0}^{5}\mathcal{A}_{i}\bigcup \{|S_2\rangle\}$ has a similar plane structure in every bipartition under the cyclic permutation. Therefore,  it is sufficient to show that the matrix $E_{BC}\propto \mathbb{I}_{BC}$.

Since $E_{BC}=E_{BC}^{\dagger}$, if $m_{ij,kl}=0$, then $m_{kl,ij}=0$ for $ij\neq kl\in \mathbb Z_d\times \mathbb Z_d$. For convenience, we divide all off-diagonal entries into two parts, called the key and the remaining off-diagonal entries. The key off-diagonal entries are $m_{ij,kl}$ for $ij\neq kl\in \mathcal S$, where $\mathcal S=\{s\hat{d},\hat{d}s,d^*d^*,\hat{d}\hat{d}\mid s\in \mathbb Z_{\hat{d}}\}$. Naturally, the remaining off-diagonal entries are $m_{pq,st}$ for $pq\neq st$, $pq\in \mathbb Z_d\times \mathbb Z_d\backslash \mathcal S$ and $st\in \mathbb Z_d\times \mathbb Z_d$.

Consider $\mathcal A_0$, $\mathcal A_1$, $\mathcal A_2$ and $|\phi_{d^*d^*}\rangle$ of $\mathcal A_5$, we get all key off-diagonal entries are zeros from Table \ref{ta2}. However, only two off-diagonal entries $m_{\frac{d^*}{2}\hat{d},\hat{d}\frac{d^*}{2}}$ and $m_{\hat{d}\frac{d^*}{2},\frac{d^*}{2}\hat{d}}$ are not available from Table \ref{ta2} when $d$ is even. In order to show that  $m_{\frac{d^*}{2}\hat{d},\hat{d}\frac{d^*}{2}}=m_{\hat{d}\frac{d^*}{2},\frac{d^*}{2}\hat{d}}=0,$ we first prove that $m_{\frac{d^*}{2}\hat{d},ij}=m_{ij,\frac{d^*}{2}\hat{d}}=0$ for $ij \in \mathbb{Z}_d\times \mathbb{Z}_d\backslash \{\mathcal{S},$ $\frac{d^*}{2}\hat{d},\hat{d}\frac{d^*}{2}\}$. We separate the argument into three cases.
\begin{figure*}[ht]
	\centering
	\includegraphics[width=13cm]{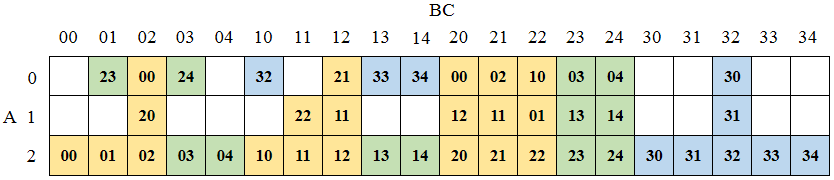}
	\caption{The $3\times 20$ plane structure of the strongest nonlocal set $\cup_{ij\in\mathbb{Z}_4\times \mathbb{Z}_5}\{|\phi_{ij}\rangle\}\bigcup \{|S_3\rangle\}$ corresponds to the bipartition $A|BC$. The yellow, green and blue entries correspond to $\mathcal S_{1}$, $\mathcal S_{2}$ and $\mathcal S_{3}$, respectively. }  \label{fig:345A}
\end{figure*}
\begin{figure}[ht]
	\centering
	\includegraphics[width=8.8cm]{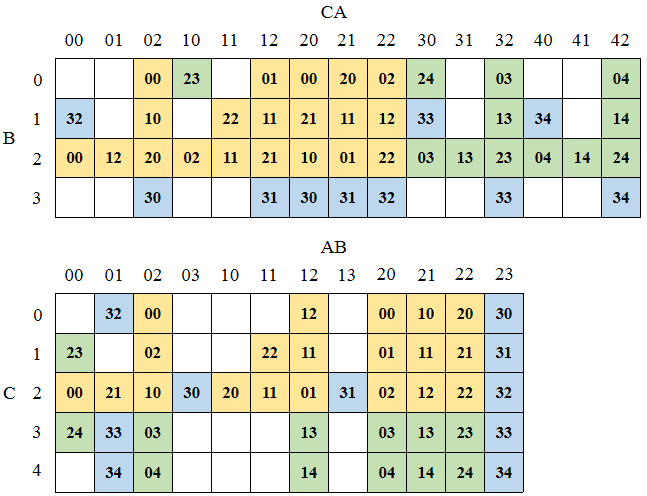}
	\caption{The $4\times 15$ and $5\times 12$ plane structures of the strongest nonlocal set $\cup_{ij\in\mathbb{Z}_4\times \mathbb{Z}_5}\{|\phi_{ij}\rangle\}\bigcup \{|S_3\rangle\}$ correspond to the bipartitions $B|CA$ and $C|AB$.  }  \label{fig:345BC}
\end{figure}
\begin{enumerate}
	\item  [\rm(i)] Neither $i$ nor $j$  is $\frac{d^*}{2}.$ Consider $|\phi_{ \frac{d^*}{2}\hat{d}} \rangle$ and $|\phi_{ij}\rangle $ by Observation \ref{ob1}, then one obtains
	$ m_{\frac{d^*}{2}\hat{d},ij}=0.$
	\item  [\rm(ii)] Only one of $i$ or $j$ is $\frac{d^*}{2}$, then $|\phi_{ij}\rangle $ belongs to $\mathcal{A}_4$ or $\mathcal{A}_5$. As $m_{\hat{d} \frac{d^*}{2},st}=0$ for $st\in \mathcal{S}\backslash \{\frac{d^*}{2}\hat{d}\}$, consider $|\phi_{ \frac{d^*}{2}\hat{d}} \rangle$ and $|\phi_{ij}\rangle $ by Observation \ref{ob2}, then one obtains
	$ m_{\frac{d^*}{2}\hat{d},ij}=0.$
		\item  [\rm(iii)] Both  of  $i$  and  $j$ are $\frac{d^*}{2}$, then $|\phi_{\frac{d^*}{2}\frac{d^*}{2}}\rangle $ belongs to $\mathcal{A}_3.$  As $m_{\hat{d} \frac{d^*}{2},\hat{d}\hat{d} }=0$, consider $|\phi_{ \frac{d^*}{2}\hat{d}} \rangle$ and $|\phi_{\frac{d^*}{2}\frac{d^*}{2}}\rangle $ by Observation \ref{ob2}, then one obtains
	$ m_{\frac{d^*}{2}\hat{d},\frac{d^*}{2}\frac{d^*}{2}}=0.$
\end{enumerate}

Next, we show that $m_{\frac{d^*}{2}\hat{d},\hat{d}\frac{d^*}{2}}=m_{\hat{d}\frac{d^*}{2},\frac{d^*}{2}\hat{d}}=0.$
Applying $|\phi _{\hat{d}\frac{d^*}{2}}\rangle$ of $\mathcal A_1$ and $|\phi _{\frac{d^*}{2}\hat{d}}\rangle$ of $\mathcal A_2$ to Eq. \eqref{eq:OR} results in:
\begin{equation}\label{eq:1}
0=\langle\phi _{\hat{d}\frac{d^*}{2}}|\mathbb{I}_{A}\otimes E_{BC}|\phi _{\frac{d^*}{2}\hat{d}}\rangle=m_{\frac{d^*}{2}\hat{d},\hat{d}\frac{d^*}{2}}+m_{\hat{d}\frac{d^*}{2},\frac{d^*}{2}\hat{d}}.
\end{equation}
As $m_{\frac{d^*}{2}\hat{d},\hat{d}\frac{d^*}{2}}$  is the complex conjugation of  $m_{\hat{d}\frac{d^*}{2},\frac{d^*}{2}\hat{d}},$ $ m_{\frac{d^*}{2}\hat{d},\hat{d}\frac{d^*}{2}}$ must be of the form
$r \sqrt{-1}$ for some real number $r$. Then, applying $|\phi _{\frac{d^*}{2}\hat{d}}\rangle$ of $\mathcal A_2$ and the stopper state $|S_2\rangle$ to Eq. \eqref{eq:OR} gives the following equation:
 \begin{equation*}
 	 0=\langle\phi _{\frac{d^*}{2}\hat{d}}|\mathbb{I}_{A}\otimes E_{BC}|S_2\rangle,
 \end{equation*}
that is,
 \begin{equation}\label{eq:2}
 0=m_{\frac{d^*}{2}\hat{d},\frac{d^*}{2}\hat{d}}-m_{\hat{d}\frac{d^*}{2},\hat{d}\frac{d^*}{2}} +m_{\frac{d^*}{2}\hat{d},\hat{d}\frac{d^*}{2}}-m_{\hat{d}\frac{d^*}{2},\frac{d^*}{2}\hat{d}}. \end{equation}
As $m_{\frac{d^*}{2}\hat{d},\frac{d^*}{2}\hat{d}}$ and $m_{\hat{d}\frac{d^*}{2},\hat{d}\frac{d^*}{2}}$ are real numbers and $ m_{\hat{d}\frac{d^*}{2},\frac{d^*}{2}\hat{d}}=-r\sqrt{-1}$, the Eq. \eqref{eq:2} implies that both the real and imaginary part are zeros. Therefore, $r=0.$  Hence, $m_{\frac{d^*}{2}\hat{d},\hat{d}\frac{d^*}{2}}=m_{\hat{d}\frac{d^*}{2},\frac{d^*}{2}\hat{d}}=0.$
Thus, all key off-dia- gonal entries $m_{ij,kl}$ are zeros, where $ij\neq kl\in \mathcal S$.

Let us show that all remaining off-diagonal entries  $m_{pq,st}$ are zeros for $pq\neq st$, $pq\in \mathbb Z_d\times \mathbb Z_d\backslash \mathcal S$ and $st\in \mathbb Z_d\times \mathbb Z_d$.  Consider $|\phi_{pq}\rangle$ and $|\phi_{st}\rangle$ by Eq. \eqref{eq:OR}, then we obtain the following equation:
$$0=\langle\phi _{pq}|\mathbb I_{A}\otimes E_{BC}|\phi _{st}\rangle= m_{pq,st} + a m_{p_aq_a,s_at_a}+b m_{p_bq_b,s_bt_b}.$$
Obviously, most of them are only composed of two terms. For $m_{p_xq_x,s_xt_x} (x\in\{a,b\}) $, the terms come from the key off-diagonal entries that are zeros. Therefore, it is not difficult to show that we can apply  Observation \ref{ob1} or Observation \ref{ob2} to  $|\phi_{pq}\rangle$ and $|\phi_{st}\rangle$, which yields $m_{pq,st}=0$.

Eventually, applying Observation \ref{ob3} to the stopper state $|S_2\rangle$ and $\mathcal A_i$ $(i\in \mathbb Z_4)$, some diagonal entries are equal, i.e., $m_{\hat{d}i,\hat{d}i}=m_{\hat{d}(d^*-i),\hat{d}(d^*-i)}=m_{i\hat{d},i\hat{d}}=m_{(d^*-i)\hat{d},(d^*-i)\hat{d}}$ for $i\in \mathbb Z_{\left\lceil\frac{d-1}{2}\right\rceil}$. Consider the stopper state $|S_2\rangle$ and $\mathcal A_4$ and $\mathcal A_5$ by Observation \ref{ob4}, we deduce that $m_{kl,kl}=m_{l\hat{d},l\hat{d}}=m_{\hat{d}k,\hat{d}k}$  for $k,l\in\mathbb{Z}_{\hat{d}}$, $k+l\neq  d^*$. Therefore, the matrix $E_{BC}\propto \mathbb{I}_{BC}$.
\end{proof}
\section{The strongest nonlocal sets with minimum cardinality in $\mathbb{C}^{d_1}\otimes \mathbb{C}^{d_2}\otimes \mathbb{C}^{d_3}$}

First of all, we construct the strongest nonlocal set of the smallest size in $\mathbb{C}^{3}\otimes \mathbb{C}^{4}\otimes \mathbb{C}^{5}$ based on the set given by Lemma~\ref{le:333}.

\begin{lemma}\label{le:345}    The set $\cup_{ij\in\mathbb{Z}_4\times \mathbb{Z}_5}\{|\phi_{ij}\rangle\}\bigcup \{|S_3\rangle\}$ of size $21$ given by Eqs. \eqref{eq:333} and \eqref{eq:345} is strongest nonlocal in $\mathbb{C}^{3}\otimes \mathbb{C}^{4}\otimes \mathbb{C}^{5}$:
\begin{equation*}\label{}
\begin{aligned}
|\phi_{03}\rangle=&|2\rangle_{A}|0\rangle_{B}|3\rangle_{C}-|0\rangle_{A}|2\rangle_{B}|3\rangle_{C},\\
|\phi_{13}\rangle=&|2\rangle_{A}|1\rangle_{B}|3\rangle_{C}-|1\rangle_{A}|2\rangle_{B}|3\rangle_{C},\\
|\phi_{04}\rangle=&|2\rangle_{A}|0\rangle_{B}|4\rangle_{C}-|0\rangle_{A}|2\rangle_{B}|4\rangle_{C},\\
|\phi_{14}\rangle=&|2\rangle_{A}|1\rangle_{B}|4\rangle_{C}-|1\rangle_{A}|2\rangle_{B}|4\rangle_{C},\\
|\phi_{23}\rangle=&|2\rangle_{A}|2\rangle_{B}|3\rangle_{C}-|0\rangle_{A}|0\rangle_{B}|1\rangle_{C},\\
|\phi_{24}\rangle=&|2\rangle_{A}|2\rangle_{B}|4\rangle_{C}-|0\rangle_{A}|0\rangle_{B}|3\rangle_{C},\\
\end{aligned}
\end{equation*}
\begin{equation}\label{eq:345}
\begin{aligned}
|\phi_{30}\rangle=&|2\rangle_{A}|3\rangle_{B}|0\rangle_{C}-|0\rangle_{A}|3\rangle_{B}|2\rangle_{C},\\
|\phi_{31}\rangle=&|2\rangle_{A}|3\rangle_{B}|1\rangle_{C}-|1\rangle_{A}|3\rangle_{B}|2\rangle_{C},\\
|\phi_{32}\rangle=&|2\rangle_{A}|3\rangle_{B}|2\rangle_{C}-|0\rangle_{A}|1\rangle_{B}|0\rangle_{C},\\
|\phi_{33}\rangle=&|2\rangle_{A}|3\rangle_{B}|3\rangle_{C}-|0\rangle_{A}|1\rangle_{B}|3\rangle_{C},\\
|\phi_{34}\rangle=&|2\rangle_{A}|3\rangle_{B}|4\rangle_{C}-|0\rangle_{A}|1\rangle_{B}|4\rangle_{C},\\
|S_3\rangle=&|0+1+2\rangle_{A}|0+1+2+3\rangle_{B}|0+1+2+3+4\rangle_{C},\\
\end{aligned}
\end{equation}
where $\cup_{i,j\in\mathbb{Z}_3}\{|\phi_{ij}\rangle\}$ are the same as the states given by Eq.~\eqref{eq:333}.
\end{lemma}

\begin{proof} Denote $\mathcal S_{1}=\cup_{i,j\in\mathbb{Z}_3}\{|\phi_{ij}\rangle\}$, $\mathcal{S}_{2}=\cup_{i\in \mathbb{Z}_3,j\in\{3,4\}}$ $\{|\phi_{ij}\rangle\}$ and $\mathcal{S}_{3}=\cup_{i\in \mathbb{Z}_5}\{|\phi_{3i}\rangle\}$. In Figs. \ref{fig:345A} and \ref{fig:345BC}, we show the plane structures of $\cup_{i=1}^{3}\mathcal{S}_{i}\bigcup \{|S_3\rangle\}$ in every bipartition.

In Lemma~\ref{le:333}, $\mathcal{S}_{1}\bigcup \{|S_1\rangle\}$ is the strongest nonlocal set in $\mathbb{C}^{3}\otimes \mathbb{C}^{3}\otimes \mathbb{C}^{3}$, this means that the matrix $E=(m_{ij,kl})_{ij, kl \in \mathbb Z_{3}\times \mathbb Z_{3}}$ is proportional to an identity matrix. Since $E=E^{\dagger}$, if $m_{ij,kl}=0$, there must be $m_{kl,ij}=0$, where $ij\neq kl\in \mathbb Z_{3}\times \mathbb Z_{3}$.

For the matrix $E_{BC}=(m_{ij,kl})_{ij, kl \in \mathbb Z_{4}\times \mathbb Z_{5}}$. Consider $\mathcal{S}_{2}$ and $\mathcal{S}_{1}$ by Observations \ref{ob1} and \ref{ob2}, like $|\phi_{23}\rangle$ and $\mathcal{S}_{1}$, we get $m_{23,ij}=0$ for $ij\in \mathbb Z_3 \times \mathbb Z_3$; from $|\phi_{03}\rangle$, $|\phi_{13}\rangle$, and $\mathcal{S}_{1}$, we have $m_{03,ij}=0$ and $m_{13,ij}=0$, where $ij\in \mathbb Z_3 \times \mathbb Z_3$; for $|\phi_{24}\rangle$ and $\mathcal{S}_{1}$, we get $m_{24,ij}=0$ for $ij\in \mathbb Z_3 \times \mathbb Z_3$; and from $|\phi_{04}\rangle$, $|\phi_{14}\rangle$, and $\mathcal{S}_{1}$, we obtain $m_{04,ij}=0$ and $m_{14,ij}=0$, where $ij\in \mathbb Z_3 \times \mathbb Z_3$. Consider $\mathcal{S}_{2}$, we get $m_{i3,j3}=m_{i4,j4}=m_{k3,l4}=0$ for $i\neq j\in \mathbb Z_3$, $k,l\in \mathbb Z_3$. Consider $\mathcal{S}_{3}$ and $\mathcal{S}_{1}$, like $|\phi_{32}\rangle$ and $\mathcal{S}_{1}$, we get $m_{32,ij}=0$ for $ij\in \mathbb Z_3 \times \mathbb Z_3$; given $|\phi_{3s}\rangle$ and $\mathcal{S}_{1}$, we have $m_{3s,ij}=0$, where $s\in \{0,1,3,4\}$, $ij\in \mathbb Z_3 \times \mathbb Z_3$. From $\mathcal{S}_{3}$ and $\mathcal{S}_{2}$, we get $m_{3i,kl}=0$ for $i\in \mathbb Z_5$, $kl\in \{j3,~j4|j\in \mathbb Z_3\}$. Consider $\mathcal{S}_{3}$, we provide $m_{3i,3j}=0$ for $i\neq j\in \mathbb Z_5$. Hereto we prove that all off-diagonal entries are zeros by Observations \ref{ob1} and \ref{ob2}, i.e., $m_{ij,kl}=0$ for $ij\neq kl$, $ij\in \{s3,~s4,~3t|s\in \mathbb Z_{3}, t\in \mathbb Z_{5}\}$, $kl\in \mathbb Z_4\times \mathbb Z_5$. Applying Observation \ref{ob3} to $\mathcal{S}_{2}$, $\mathcal{S}_{3}$, and $|S\rangle$, we obtain that all diagonal entries are equal, i.e., $m_{ij,ij}=m_{kl,kl}$ for $ij\neq kl\in \mathbb Z_4\times \mathbb Z_5$. Thus, the matrix $E_{BC}\propto \mathbb{I}_{BC}$.

Consider the matrix $E_{CA}=(m_{ij,kl})_{ij, kl \in \mathbb Z_{5}\times \mathbb Z_{3}}$. Applying Observations \ref{ob1} and \ref{ob2} to $\mathcal{S}_{2}$ and $\mathcal{S}_{1}$, like $|\phi_{23}\rangle$ and $\mathcal{S}_{1}$, we get $m_{32,ij}=0$ for $ij\in \mathbb Z_3 \times \mathbb Z_3$; from $|\phi_{03}\rangle$, $|\phi_{13}\rangle$, and $\mathcal{S}_{1}$, we have $m_{30,ij}=0$ and $m_{31,ij}=0$ for $ij\in \mathbb Z_3 \times \mathbb Z_3$; given $|\phi_{24}\rangle$ and $\mathcal{S}_{1}$, we get $m_{42,ij}=0$ for $ij\in \mathbb Z_3 \times \mathbb Z_3$; and from $|\phi_{04}\rangle$, $|\phi_{14}\rangle$, and $\mathcal{S}_{1}$, we obtain $m_{40,ij}=0$ and $m_{41,ij}=0$ for $ij\in \mathbb Z_3 \times \mathbb Z_3$. Consider $\mathcal{S}_{2}$, we have $m_{3i,3j}=m_{4i,4j}=m_{3k,4l}=0$ for $i\neq j\in \mathbb Z_3$, $k,l\in \mathbb Z_3$. Applying Observation \ref{ob3} to $\mathcal{S}_{2}$ and $|S_3\rangle$, we present that $m_{3i,3i}=m_{4i,4i}=m_{10,10}$ for $i\in \mathbb Z_3$. So, the matrix $E_{CA}$ is proportional to an identity matrix.

For the matrix $E_{AB}=(m_{ij,kl})_{ij, kl \in \mathbb Z_{3}\times \mathbb Z_{4}}$. Consider states $|\phi_{30}\rangle$, $|\phi_{31}\rangle$, $|\phi_{32}\rangle$, and $\mathcal{S}_{1}$, we get $m_{i3,kl}=0$ for $i\in \mathbb Z_3$, $kl\in \mathbb Z_3\times \mathbb Z_3$.
Applying Observations \ref{ob1} and \ref{ob2} to $|\phi_{30}\rangle$, $|\phi_{31}\rangle$, and $|\phi_{32}\rangle$, we yield $m_{i3,j3}=0$ for $i\neq j\in \mathbb Z_3$. From $|\phi_{30}\rangle$, $|\phi_{31}\rangle$, $|\phi_{32}\rangle$, and $|S_3\rangle$ by Observation \ref{ob3}, we have $m_{03,03}=m_{13,13}=m_{23,23}=m_{01,01}$. Thus the POVM element $E_{AB}$ is trivial, which completes the proof.
\end{proof}

In Ref.~\cite{Li2023PRA}, Li {\it et al}. proposed the strongest nonlocal set of size $22$ in $\mathbb{C}^{3}\otimes \mathbb{C}^{4}\otimes \mathbb{C}^{5}$. Here, we construct the strongest nonlocal set of the smallest size $21$, which reaches the low bound on the strongest nonlocal sets in $\mathbb{C}^{3}\otimes \mathbb{C}^{4}\otimes \mathbb{C}^{5}$.

\begin{theorem} \label{th:d1d2d3}  The set $\cup_{i=0}^{13}\mathcal{A}_{i}\bigcup \{|S\rangle\}$ of size $d_{2}d_{3}+1$ given by Eqs. \eqref{eq:ddd} and \eqref{eq:d1d2d3} is strongest nonlocal in $\mathbb{C}^{d_{1}}\otimes$ $\mathbb{C}^{d_2}\otimes \mathbb{C}^{d_3}$ for $3\leq d_1\leq d_2\leq d_3$:
\begin{widetext}
\begin{equation}\label{eq:d1d2d3}
\begin{aligned}
\mathcal{A}_{6}=&\{|\phi_{i(d_{1}+j)}\rangle=|\hat{d_{1}}\rangle_{A}|i\rangle_{B}|d_{1}+j\rangle_{C}-|i\rangle_{A}|\hat{d_{1}}\rangle_{B}|d_{1}+j\rangle_{C}\mid i\in \mathbb Z_{\hat{d_{1}}}, j\in \mathbb Z_{d_{3}-d_{1}}\},\\
\mathcal{A}_{7}=&\{|\phi_{\hat{d_{1}}d_{1}}\rangle=|\hat{d_{1}}\rangle_{A}|\hat{d_{1}}\rangle_{B}|d_{1}\rangle_{C}-|0\rangle_{A}|0\rangle_{B}|1\rangle_{C}\},\\
\mathcal{A}_{8}=&\{|\phi_{\hat{d_{1}}(d_{1}+1+i)}\rangle=|\hat{d_{1}}\rangle_{A}|\hat{d_{1}}\rangle_{B}|d_{1}+1+i\rangle_{C}-|0\rangle_{A}|0\rangle_{B}|d_{1}+i\rangle_{C}\mid i\in \mathbb Z_{d_{3}-d_{1}-1}\},\\
\mathcal{A}_{9}=&\{|\phi_{(d_{1}+j)i}\rangle=|\hat{d_{1}}\rangle_{A}|d_{1}+j\rangle_{B}|i\rangle_{C}-|i\rangle_{A}|d_{1}+j\rangle_{B}|\hat{d_{1}}\rangle_{C}\mid i\in \mathbb Z_{\hat{d_{1}}}, j\in \mathbb Z_{d_{2}-d_{1}}\},\\
\mathcal{A}_{10}=&\{|\phi_{d_{1}\hat{d_{1}}}\rangle=|\hat{d_{1}}\rangle_{A}|d_{1}\rangle_{B}|\hat{d_{1}}\rangle_{C}-|0\rangle_{A}|1\rangle_{B}|0\rangle_{C}\},\\
\mathcal{A}_{11}=&\{|\phi_{(d_{1}+1+i)\hat{d_{1}}}\rangle=|\hat{d_{1}}\rangle_{A}|d_{1}+1+i\rangle_{B}|\hat{d_{1}}\rangle_{C}-|0\rangle_{A}|d_{1}+i\rangle_{B}|0\rangle_{C}\mid i\in \mathbb Z_{d_{2}-d_{1}-1}\},\\
\mathcal{A}_{12}=&\{|\phi_{d_{1}(d_{1}+i)}\rangle=|\hat{d_{1}}\rangle_{A}|d_{1}\rangle_{B}|d_{1}+i\rangle_{C}-|0\rangle_{A}|1\rangle_{B}|d_{1}+i\rangle_{C}\mid i\in \mathbb Z_{d_{3}-d_{1}}\},\\
\mathcal{A}_{13}=&\{|\phi_{(d_{1}+1+i)(d_{1}+j)}\rangle=|\hat{d_{1}}\rangle_{A}|d_{1}+1+i\rangle_{B}|d_{1}+j\rangle_{C}-|0\rangle_{A}|d_{1}+i\rangle_{B}|d_{1}+j\rangle_{C}\mid i\in \mathbb Z_{d_{2}-d_{1}-1}, j\in \mathbb Z_{d_{3}-d_{1}}\},\\
|S\rangle=&\left(\sum_{i\in \mathbb Z_{d_1}}|i\rangle_{A}\right)\left(\sum_{j\in \mathbb Z_{d_2}}|j\rangle_{B}\right)\left(\sum_{k\in \mathbb Z_{d_3}}|k\rangle_{C}\right),\\
\end{aligned}
\end{equation}
\end{widetext}
where $\hat{d_1}=d_{1}-1$, and $\mathcal{A}_{i}$ $(i\in \mathbb Z_6)$ are similar to the sets given by Eq.~\eqref{eq:ddd}, except that $d$ is replaced by $d_{1}$.
\end{theorem}

\begin{proof}
Denote that $\mathcal{B}_{1}=\cup_{i=0}^{5}\mathcal{A}_{i}$, $\mathcal{B}_{2}=\cup_{i=6}^{8}\mathcal{A}_{i}$, and $\mathcal{B}_{3}=\cup_{i=9}^{13}\mathcal{A}_{i}$. In Theorem~\ref{th:ddd}, $\mathcal{B}_{1}\bigcup \{|S\rangle\}$ is the strongest nonlocal set in $\mathbb{C}^{d_1}\otimes \mathbb{C}^{d_1}\otimes \mathbb{C}^{d_1}$, it implies that $m_{ij,kl}=0$ and $m_{ij,ij}=m_{kl,kl}$ for $ij\neq kl \in \mathbb Z_{d_1}\times \mathbb Z_{d_1}$. In order to show that the matrices $E_{BC}$, $E_{CA}$, and $E_{AB}$ are trivial, we divide all off-diagonal entries into two parts in every bipartition, i.e., the key and the remaining off-diagonal entries.

In the bipartition $A|BC$, the key off-diagonal entries are $m_{ij,kl}$ for $ij\neq kl\in \mathcal K_{1}$, where $\mathcal K_1=\{d_{1}^{*}d_{1}^{*}, i\hat{d_1}, \hat{d_1}j\mid d_{1}^{*}=d_1-2, i\in \mathbb Z_{d_2}, j\in \mathbb Z_{d_3}\}$; the remaining off-diagonal entries in the matrix $E_{BC}$ are $m_{pq,st}$ for $pq\neq st$, $pq\in \mathbb Z_{d_2}\times \mathbb Z_{d_3}\backslash \mathcal K_{1}$, and $st\in \mathbb Z_{d_2}\times \mathbb Z_{d_3}$. Applying Observations \ref{ob1} and \ref{ob2} to the sets $\mathcal{B}_{1}$, $\mathcal{A}_{7}$, $\mathcal{A}_{8}$, $\mathcal{A}_{10}$ and $\mathcal{A}_{11}$, we get that all key off-diagonal entries are zeros. Since $m_{ij,kl}=m_{kl,ij}=0$ for $ij\neq kl\in \mathcal K_{1}$, applying Observation \ref{ob1}, \ref{ob2} to $\mathcal{B}_{1}$, $\mathcal{B}_{2}$, and $\mathcal{B}_{3}$, we obtain that all off-diagonal entries are zeros. Consider $\mathcal{B}_{1}$, $\mathcal{B}_{2}$, $\mathcal{B}_{3}$, and $|S\rangle$ by Observations \ref{ob3} and \ref{ob4}, and we have $m_{ij,ij}=m_{kl,kl}$ for $ij\neq kl\in \mathbb Z_{d_2}\times \mathbb Z_{d_3}$. Thus, the matrix $E_{BC}$ is trivial.

Consider the bipartition $B|CA$, the key off-diagonal entries are $m_{ij,kl}$ for $ij\neq kl\in \mathcal K_{2}$, where $\mathcal K_2=\{d_{1}^{*}d_{1}^{*},$ $j\hat{d_1} \mid d_{1}^{*}=d_1-2, j\in \mathbb Z_{d_3}\}$; the remaining off-diagonal entries in the matrix $E_{CA}$ are $m_{pq,st}$ for $pq\neq st$, $pq\in \mathbb Z_{d_3}\times \mathbb Z_{d_1}\backslash \mathcal K_{2}$ and $st\in \mathbb Z_{d_3}\times \mathbb Z_{d_1}$. Applying Observations \ref{ob1} and \ref{ob2} to $\mathcal{B}_{1}$, $\mathcal{A}_{7}$, and $\mathcal{A}_{8}$, we obtain that $m_{ij,kl}=m_{kl,ij}=0$ for $ij\neq kl\in \mathcal K_{2}$. Based on the fact that all key off-diagonal entries are zeros, applying Observations \ref{ob1} and \ref{ob2} to $\mathcal{B}_{1}$ and $\mathcal{B}_{2}$, we have $m_{ij,kl}=m_{kl,ij}=0$ for $ij\neq kl\in \mathbb Z_{d_3}\times \mathbb Z_{d_1}$. From $\mathcal{B}_{1}$, $\mathcal{B}_{2}$, and $|S\rangle$ by applying Observations \ref{ob3} and \ref{ob4}, it implies that $m_{ij,ij}=m_{kl,kl}$ for $ij\neq kl\in \mathbb Z_{d_3}\times \mathbb Z_{d_1}$. Therefore, the matrix $E_{CA}\propto \mathbb{I}_{CA}$.

For the bipartition $C|AB$, the key off-diagonal entries are $m_{ij,kl}$ for $ij\neq kl\in \mathcal K_{3}$, where $\mathcal K_3=\{d_{1}^{*}d_{1}^{*}, \hat{d_1}i \mid d_{1}^{*}=d_1-2, i\in \mathbb Z_{d_2}\}$; the remaining off-diagonal entries in the matrix $E_{AB}$ are $m_{pq,st}$ for $pq\neq st$, $pq\in \mathbb Z_{d_1}\times \mathbb Z_{d_2}\backslash \mathcal K_{3}$, and $st\in \mathbb Z_{d_1}\times \mathbb Z_{d_2}$. According to $\mathcal{B}_{1}$, $\mathcal{A}_{10}$ and $\mathcal{A}_{11}$ by applying Observations \ref{ob1} and \ref{ob2}, we obtain that all key off-diagonal entries are zeros. Since $m_{ij,kl}=m_{kl,ij}=0$ for $ij\neq kl\in \mathcal K_{3}$, applying Observations \ref{ob1} and \ref{ob2} to $\mathcal{B}_{1}$ and $\mathcal{B}_{3}$, we get that all off-diagonal entries are zeros. Considering $\mathcal{B}_{1}$, $\mathcal{B}_{3}$ and $|S\rangle$ by Observations \ref{ob3} and \ref{ob4}, we deduce $m_{ij,ij}=m_{kl,kl}$ for $ij\neq kl\in \mathbb Z_{d_1}\times \mathbb Z_{d_2}$. Hence, the matrix $E_{AB}\propto \mathbb{I}_{AB}$. To sum up, we successfully show that the set $\cup_{i=1}^{3}\mathcal{B}_{i}\bigcup \{|S\rangle\}$ is the strongest nonlocal set in general tripartite systems.
\end{proof}

\section{Conclusion}

In this paper, we constructed the strongest nonlocal sets with minimum cardinality in general tripartite systems. Our result positively answers an open conjecture proposed in Ref.~\cite{Li2023}. From Table~\ref{results}, the size of the strongest nonlocal sets we constructed is the smallest of all previous results in $\mathbb{C}^{d_1}\otimes \mathbb{C}^{d_2}\otimes \mathbb{C}^{d_3}$. In particular, the strongest nonlocal set given by Theorem~\ref{th:ddd} consists of $d^2$ orthogonal genuine entangled states except the stopper state, and the strongest nonlocal set given by Theorem~\ref{th:d1d2d3} contains $d_2d_3$ orthogonal entangled states except the stopper state.

There are still some open questions left to be solved. In general $N$-partite systems $\otimes_{i=1}^{N}\mathbb{C}^{d_i}$, can we construct the strongest nonlocal set of the smallest size $\max_{i}\{\prod_{j=1}^{n}d_j/d_i+1\}$? Can we further improve the lower bound on the strongest nonlocal sets?
\section*{Acknowledgments}

This work is supported by the Natural Science Foundation of Hebei Province under Grant No. F2021205001, the National Natural Science Foundation of China under Grants No. 62272208 and No. 12371458 and the Guangdong Basic and Applied Basic Research Foundation under Grant
No. 2023A1515012074.

\end{document}